\documentclass[sigconf,nonacm, anonymous=False, review=False]{acmart}

\renewcommand\footnotetextcopyrightpermission[1]{}
\setcopyright{none}

\usepackage{graphicx,subcaption}
\usepackage{todonotes}
\usepackage{amsmath}
\usepackage{svg}
\usepackage{amsfonts}
\usepackage{multirow}
\usepackage{amsthm,hyperref,cleveref}
 
\usepackage{bm}
\usepackage{fancyhdr}

\newtheorem{thm}{Theorem}[section]
\newtheorem{lem}[thm]{Lemma}
\newtheorem{cor}[thm]{Corollary}

\newtheorem{observation}{Observation}
\usepackage{graphbox}

\usepackage[symbol]{footmisc}

\newcommand{\matr}[1]{\bm{#1}}     %

\usepackage{mathtools}
\DeclarePairedDelimiter{\ceil}{\lceil}{\rceil}
\DeclarePairedDelimiter{\floor}{\lfloor}{\rfloor}

\DeclareMathOperator{\E}{\mathbb{E}}

\begin{document}

\title[The spatial computer]{The spatial computer: \\ A model for energy-efficient parallel computation}
\date{}

\author{Lukas Gianinazzi*, Tal Ben-Nun, Maciej Besta, Saleh Ashkboos, \\ Yves Baumann, Piotr Luczynski, Torsten Hoefler}       
 \affiliation{\department{Department of Computer Science} \institution{ETH Zurich} \country{Switzerland}}

\renewcommand{\shortauthors}{Gianinazzi et al.}

\begin{abstract}
We present a new parallel model of computation suitable for \emph{spatial} architectures, for which the energy used for communication heavily depends on the distance of the communicating processors. %
In our model, processors have locations on a conceptual two-dimensional grid, and their distance therein determines their communication cost. In particular, we introduce the \emph{energy} cost of a spatial computation, which measures the total distance traveled by all messages, and study the \emph{depth} of communication, which measures the largest number of hops of a chain of messages. 
We show matching energy lower- and upper bounds for many foundational problems, including sorting, median selection, and matrix multiplication. Our model does not depend on any parameters other than the input shape and size, simplifying algorithm analysis. We also show how to simulate PRAM algorithms in our model and how to obtain results for a more complex model that introduces the size of the local memories of the processors as a parameter.

\end{abstract}


\maketitle

\footnotetext[1]{Corresponding Author: lukas.gianinazzi@inf.ethz.ch}

\section{Introduction}

\def\arraystretch{1.7} 
\setlength{\tabcolsep}{0.78em} 
\begin{table*}[t]
  \centering
  \begin{tabular}{llllllcc}
    Input Shape & Problem & & Energy & Depth & Wire-Depth & Data-Oblivious & Deterministic\\ \hline
    $\sqrt{n}\times\sqrt{n}$ & Broadcast / Reduce  & & $\Theta(n)$ & $O(\log n)$ & $O(\sqrt{n})$ & \checkmark &  \checkmark  \\
    $\sqrt{n}\times\sqrt{n} \enspace {}^{(\dag)}$  & Parallel Scan& &$\Theta(n)$ & $O(\log n)$ & $O(\sqrt{n})$ &  \checkmark &\checkmark \\
    $\sqrt{n}\times\sqrt{n}$ & Rank Selection & & $\Theta(n)$ & $O(\log ^ 2 n)$ & $O(\sqrt{n})$ & - & -  \\
    $\sqrt{n}\times\sqrt{n}$ & Sorting & & $\Theta\left(n^{\frac{3}{2}}\right)$ & $O(\log^3 n)$ & $O(\sqrt{n})$ & - &  \checkmark  \vspace{0em}\\ 
    \multirow{2}{*}{${n}\times{n}$} & \multirow{2}{*}{Matrix Multiplication}  & Strassen & $O\left(n^{3} \sqrt{\log n} \right)$ & $O(n^{0.861})$ & $O(n \sqrt{\log n})$ &  \checkmark &  \checkmark \\
    &   & Cubic & $\Theta(n^3)$ & $\Theta(n)$ & $O(n)$ &  \checkmark &  \checkmark \\
  \end{tabular}
  \vspace{0.5em}
  \caption{Summary of energy and depth bounds on the spatial computer. We provide tight energy upper and lower bounds for parallel scan, sorting, rank selection, and matrix multiplication. 
  $(^{\dag})$ Input and output stored in Z-order.}
  \vspace{-1em}
  \label{tab:bounds}
\end{table*}

Energy consumption has become a major economic and technical factor in parallel systems~\cite{DBLP:journals/micro/KecklerDKGG11, DBLP:journals/corr/abs-2012-03112, DBLP:journals/cse/KoggeS13}. For many chips, such as GPUs, data movement is one of the main drivers behind non-idle energy consumption~\cite{DBLP:journals/micro/KecklerDKGG11}. Communicating with a distant off-chip memory module is several orders of magnitude more costly than communicating within the same chip and more generally, the closer the memory the less energy is required to access it~\cite{DBLP:journals/cse/KoggeS13}.
Under these circumstances, always needing to communicate over a monolithic shared memory is inefficient. 
This has driven hardware architects towards designs in which many cores have their own local memories and can communicate directly with each other. As this leads to closer cores being able to communicate faster and with less energy, this leads to a \emph{spatial computing architecture}. 
Early approaches are systolic architectures~\cite{DBLP:journals/computer/Kung82}, while modern examples include Cerebras Wafer-Scale Engine (WSE)~\cite{Cerebras}, hierarchical many-cores~\cite{DBLP:journals/micro/ZarubaSB21, DBLP:conf/date/CavalcanteRPB21}, and Coarse-Grained Reconfigurable Architectures (CGRA)~\cite{DBLP:conf/asap/ChinSRZKHA17, DBLP:conf/fpga/GaideGRB19}.

The standard computation and communication complexity models do not suffice for spatial hardware. In the shared memory setting, previous approaches to reduce data movement have focused on the use of caches to maximize data reuse and reduce the need for communication~\cite{DBLP:conf/ipps/ArgeGS10}. In the distributed memory setting, the goal has been to reduce the total communication volume~\cite{DBLP:journals/cacm/Valiant90}. However, larger caches provide diminishing returns in the context of spatial hardware~\cite{DBLP:journals/corr/abs-2012-03112} and in practice the cost of communicating data between processing elements \emph{depends on their distance}~\cite{DBLP:journals/micro/KecklerDKGG11, DBLP:journals/cse/KoggeS13, DBLP:conf/date/CavalcanteRPB21}. 

We propose to explicitly model the non-uniform communication costs on large chips, where each processor has its own local memory. In our model, the \emph{energy} consumed for communication is proportional to the physical distance of the communicating processors. We do not assume a fixed communication topology, but instead provide parallelism bounds that translate to different physical realizations of the communication network on the chip. As such, our model serves as a new \emph{bridging model} between algorithms and contemporary parallel computers. Because our model heavily encourages spatial communication
 locality, we can determine the effect of larger local memories without modeling them explicitly.

The new setting poses two particularly interesting algorithmic challenges.
(1) \emph{Data Layout}. Because the distance of the communicating processors is crucial for our energy cost, the \emph{data layout} (in 2D space) becomes a central concern for algorithms and data structure design. This becomes exponentially important for recursive algorithms.
(2) \emph{Permutation Bottleneck}. We show that permutation (all-to-all communication) requires much more than linear energy. Hence, it becomes the bottleneck in many classic algorithm designs. To obtain near-linear energy bounds, permutation (or sorting) of the whole input must be avoided and different techniques applied.

To address these challenges, we present several techniques. To address the data-layout challenge, \emph{space filling curves} provide energy-optimal layouts that enable highly parallel computation. Moreover, \emph{unrolling recursive algorithms} a constant number of times provides new opportunities for their efficient 2D layout. To overcome the permutation bottleneck, we propose using \emph{sampling} to avoid expensive all-to-all communication. If a small sample can be used to progress the original problem towards its solution very quickly, we obtain a linear-energy algorithm.

\subsection{Trends in Parallel Computer Architecture}


The development of new compute architecture paradigms is driven by fundamental shifts in hardware design constraints~\cite{DBLP:journals/cse/KoggeS13}. Most importantly, power consumption becomes the limiting constraint for high performance systems. The high amounts of consumed power are not only due to the increasing scales of developed architectures, but also due to the fact that data movement is very power hungry, and many modern applications and workloads, such as large-scale deep learning~\cite{ben2019demystifying, ben2019modular, besta2022parallel}, are dominated by data movement.
%
%
To address these challenges, fundamentally new hardware designs have been proposed~\cite{DBLP:conf/date/CavalcanteRPB21, DBLP:journals/micro/ZarubaSB21, DBLP:conf/cvpr/FarabetMCACL11, besta2018slim, xu2017benchmarking, liu2021closing}. Despite their differences, we argue that spatial locality is the common key to design successful algorithms for these architectures.

\emph{Hierarchical many-core} architectures organize hundreds or thousands of simple compute cores in a hierarchical fashion~\cite{DBLP:conf/date/CavalcanteRPB21, DBLP:journals/micro/ZarubaSB21}. Conceptually, the cores are placed in a semi-regular grid onto the chip and a tree is built on top of the cores. Siblings in this tree can communicate with high bandwidth and low latency. However, the communication bandwidth to the parents decreases towards the root of the hierarchy. This is necessary to avoid using too much fabric for the communication wires. Each node in the hierarchy has its own storage that can be accessed similarly to a shared memory, however with highly non-uniform cost. 
Examples of hierarchical many-cores include Mempool~\cite{DBLP:conf/date/CavalcanteRPB21}, a chip with 256 cores organized in a 3-level hierarchy, and Manticore~\cite{DBLP:journals/micro/ZarubaSB21}, a chip with 4,096 cores and a 5-level hierarchy.

We observe that for these systems, it is beneficial to reduce the volume of the messages that occur between distant cores to avoid congestion due to the limited bandwidth available between distant cores. By penalizing communication by its spatial distance, we can incentivize communication patterns that reduce communication across \emph{any levels of the hieararchy} (see \Cref{sec:hierarchy}). Moreover, we observe that because of the hierarchical tree-like interconnects, the latency between far-away cores is relatively small, meaning that we want to incentivize highly parallel algorithms.





\emph{Reconfigurable dataflow} architectures consist of a 2D grid of reconfigurable units connected by a mesh-like network~\cite{DBLP:conf/cvpr/FarabetMCACL11, DBLP:journals/cse/EmaniVAPSFJLNSK21, DBLP:conf/fpga/GaideGRB19, besta2019graph}. Communication to distant units is possible either through a separate network~\cite{DBLP:conf/fpga/GaideGRB19} or by going through multiple hops in the mesh network~\cite{DBLP:journals/cse/EmaniVAPSFJLNSK21, Cerebras}. 
%
%
%
Examples include the SambaNova Reconfigurable Dataflow Architecture~\cite{DBLP:journals/cse/EmaniVAPSFJLNSK21}, NeuFlow~\cite{DBLP:conf/cvpr/FarabetMCACL11}, Cerebras WSE~\cite{Cerebras}, and the Xilinx ACAP Versal~\cite{DBLP:conf/fpga/GaideGRB19}.
We observe that these networks have high latency to far away nodes and high bandwidth to close nodes. This means we want to incentivize highly local communication and reduce the physical length of the critical path.

\subsection{The Spatial Computer}

Our goal is to design a highly productive model for algorithm designers that abstracts away the details of a specific hardware architecture, but incentivizes (a) spatially local communication and (b) large amounts of parallelism. If our goal was only (a), we would get a mesh-like architecture~\cite{DBLP:conf/spaa/KaufmannRS92, DBLP:conf/parle/LeppanenP94}. If our goal was only (b), we would get an architecture where all communication cost are the same for all pairs of processors~\cite{ReifSynthesis, DBLP:journals/cacm/Valiant90}. 
Next, we introduce our new model of computation and communication cost model, which models the key aspects of a \emph{spatial computer}.

\paragraph{Computer Model}

We consider an unbounded number of processors organized in a regular (Cartesian) 2-dimensional grid. This grid does not restrict the ability of processors to communicate with each other, but it determines the cost of doing so (c.f. \emph{cost model}). 
Each processor $p_{i, j}$ knows its position $(i, j)$ in the grid and has a constant-sized memory. During each synchronous time-step, each processor proceeds as follows. It can send a constant-unit-sized message to an arbitrary other processor. This message arrives in the next time step. Then, a processor can dequeue a message from its receive queue and put it into its memory. Finally, a processor can perform a constant number of arithmetic and logic operations on its memory, and generate an independent, uniformly distributed word-sized number. These operations determine which message to send next.
Note that the order in which messages arrive in the receive queues is not deterministic. The receive queues are of constant size. If a processor receives more than this constant number of messages in a time-step, the behavior is undefined.

Initially, an input of size $n$ is distributed in some predefined format on an $h\times w$ sized processor \emph{subgrid} containing the processor at $(0, 0)$, where $n=\Theta(hw)$. For simplicity, we assume w.l.o.g. that $n$ is a power of $4$. We sometimes refer to the processor $p_{0, 0}$ as the \emph{root} processor. 
For each problem, the output format is also predefined.

\paragraph{Cost Model}

Sending a message from processor $p_{i, j}$ in the grid to processor $p_{x, y}$ in the grid costs $|x-i|+|y-j|$ \emph{energy}. The energy $E$ of a computation is the sum of the energy cost of the sent messages. 

If a message $m'$ is sent by a processor after receiving a message $m$ we say that $m'$ \emph{depends on} $m$. The longest chain of messages that consecutively depend on each other is the \emph{depth} $D$ of the computation. 
The largest total energy of any chain of messages that consecutively depend on each other is the \emph{wire-depth} $D_w$. We obtain a simple upper bound on the wire-depth based on the diameter of the utilized subgrid and the depth: an algorithm that uses a $h\times w$ compute grid and has depth $D$ has wire-depth at most $D_w\leq D (w + h)$. Our goal is to improve on this trivial upper bound.

A bound $f(n)\in O(g(n))$ holds with high probability (w.h.p.), if for any constant $d$, there exist constants $n_0>0$ and $c>0$ such that for all $n\geq n_0$, $f(n)\leq c g(n)$ with probability at least $1-n^{-d}$.

\subsection{Related Work} 

\paragraph{Fixed-Connection Network Model} The fixed-con\-nec\-tion network model~\cite{Leighton1991IntroductionTP} considers computation for several particular communication topologies, such as toruses, hypercubes, and meshes. In contrast, we do not restrict the topology of the communication network, but instead derive bounds that are meaningful for a variety of concrete network topologies. Moreover, we explicitly consider the distance traveled as a metric (mainly for energy reasons), whereas Leighton's fixed-connection network ignores the wire-lengths~\cite{Leighton1991IntroductionTP}.

A related set of works focuses on computation on mesh networks~\cite{DBLP:conf/spaa/KaufmannRS92, DBLP:conf/parle/LeppanenP94, DBLP:journals/njc/LeppanenP95, DBLP:conf/dimacs/GoddardKP94}.  Note that any algorithm on a mesh network could be executed in our model: An algorithm that takes $L$ rounds on a $w\times w$ mesh network takes $O(L w^2)$ energy, has depth $L$ and wire-depth $O(L)$. However, since for many problems no sub-polynomial round algorithms are possible on a mesh-interconnect, the algorithmic approaches for these models differ significantly from ours.

\paragraph{Very Large Scale Integration (VLSI)} VLSI complexity~\cite{DBLP:journals/tc/Thompson83a, DBLP:journals/jacm/ChazelleM85, DBLP:journals/jcss/Savage81, DBLP:conf/stoc/LiptonS81} considers a network wiring that satisfies certain geometric constraints on how densely gates can be placed and how many wire-crossings are allowed. Then, the main concerns for such VLSI hardware models are area-depth-product tradeoffs of the circuitry~\cite{Leighton2003ComplexityII}. For calculation of the depth, the length of the wires is ignored~\cite{DBLP:journals/tc/Thompson83a}. 
%
 %
In contrast, we model a spatial computer consisting of general purpose cores whose communication patterns can depend on the data and are routed through an on-chip network in a \emph{network-oblivious} way. Thus, our model is intended to have a higher level of abstraction and is geared towards algorithm design (rather than circuit design). Moreover, we take the communication distance into account when measuring costs.

\paragraph{Work/Depth and PRAM} The work-depth model considers computation as a directed acyclic graph and counts the total number of operations (work) and the length of the critical path (depth). Up to constant factors, the depth in the work-depth model is the same as the depth in our model. 
The work and depth model is closely related to the PRAM model. The PRAM~\cite{ReifSynthesis} assumes a single shared memory with \emph{uniform cost random access}. An algorithm with work $W$ and depth $D'$ takes $O(W/p+D')$ time on $p$ PRAM processors. We will show how to simulate any PRAM algorithm in our model. Usually, this will not lead to energy-optimal algorithms, as the random accesses in PRAM might result in a full permutation of the data in each time step.

\paragraph{PCRAM} Hora et al.~\cite{DBLP:conf/tamc/HoraKT19} present an abstract model for FPGAs as accelerators. They assume a simplified model where the FPGA runs execution DAGs in a pipelined manner, abstracting away the complexities of configuring circuits on the FPGA. The FPGA is connected to a large random access memory. We take a fundamentally different point of view, as we consider the spatial aspects of the hardware and heavily penalize random access.

\paragraph{BSP} In the bulk-synchronous parallel model~\cite{DBLP:journals/cacm/Valiant90,DBLP:conf/spaa/AdlerDJKR98, DBLP:conf/spaa/GerbessiotisS96, DBLP:conf/ipps/GerbessiotisS97} (BSP), processors can communicate arbitrarily large messages in synchronous rounds, called supersteps. 
The goal is to reduce the number of rounds, the communication volume, and finally the computation time. The model is well-suited to a coarse-grained parallel scenario (where the number of processors is much smaller than the input size and each processor has a large local memory), which occurs when connecting many large machines into a cluster. In BSP, communication is assumed to have the same costs for each pair of processors. In contrast, in our model the cost between communicating processors is spatially dependent.

\paragraph{PIM Model} Kang et al.~\cite{DBLP:conf/spaa/KangGBD0M21} present a model for processing-in-memory (PIM) that combines aspects of a PRAM with aspects of a BSP machine. A set of processing-in-memory cores are treated as a distributed system that acts as an accelerator to a PRAM-like set of host cores. This model leads to hierarchical solution approaches, but treats all PIM cores and host processors as symmetric to each other and does not consider any spatial aspects. 

\subsection{Discussion}

Compared to previous communication cost models~\cite{DBLP:journals/cacm/Valiant90, DBLP:journals/cacm/CullerKPSSSSE96}, we weight communication cost by the distance traveled. In addition to improving energy-consumption, sending messages locally decreases congestion for links that connect far-away parts of the system, as is important for hierarchical many-cores. An algorithm with energy $E$ sends at most $E/k$ messages to processors that have distance $k$ or more. This bound is nontrivial if the energy is less than $\sqrt{n}$ times the number of messages sent, which is the case for all our energy-optimal algorithms. In particular, an $O(n)$ energy algorithm will send only $O(n/k)$ messages farther than distance $k$. 
This means that the bandwidth required to avoid congestion decreases linearly with the distance.

Depth ignores the length of the path traveled. This is a good estimate of the latency when the underlying communication network on the chip~\cite{DBLP:conf/dac/DallyT01} has a low diameter, such as in hierarchical many-cores. 
%
Reconfigurable dataflow architectures use mesh-like network topologies~\cite{DBLP:conf/cvpr/FarabetMCACL11, DBLP:journals/cse/EmaniVAPSFJLNSK21, DBLP:conf/fpga/GaideGRB19}, which use little chip area~\cite{DBLP:journals/taco/SanchezMK10, DBLP:journals/csur/BjerregaardM06} and are thus more economical. However, the latency grows with the distance of the communicating units. Hence, it is desirable to consider wire-depth as well. Our algorithms optimize spatial locality in terms of both energy and wire-depth.

\subsection{Our Contribution}

For several foundational algorithmic problems, we present energy-optimal, low-depth, and wire-depth optimal algorithms in the spatial computer model.
We show an energy-lower bound for the permutation problem. In particular, permuting $n$ numbers on a square grid takes $\Omega(n^{3/2})$. Since sorting and square matrix multiplication can be used to implement energy-intensive permutations, the energy lower bounds for these problems follow. The permutation lower bound means that many classic techniques which rely on permuting the data do not lead to efficient algorithms.

Instead, our algorithms exploit two-dimensional recursive patterns to increase their communication locality, decreasing energy cost and wire-depth. We provide communication primitives that are optimal on subgrids that are not exponentially taller than wide. Next, we study two comparison-based problems, rank-selection, and sorting. In light of the permutation energy lower-bound, it is remarkable that rank selection can be achieved in linear $O(n)$ energy and low depth. Moreover, we provide energy-depth tradeoffs for matrix multiplication. See \Cref{tab:bounds} for a summary of our bounds. 

We then show how to simulate existing PRAM algorithms using our energy-optimal sorting algorithm. This simulation can be helpful to implement parts of an algorithm that are not the energy bottleneck.
Finally, we explore a model variant where processors possess a larger local memory and can receive more than a constant number of messages. We show that because our model weights communication costs by the distance traveled, there is a direct translation between results for constant-sized local memories and larger local memories. This simulation result means we can focus our attention on designing algorithms for the simpler model with constant-sized memories.

\section{Communication Primitives}

We start by modeling communication primitives consisting of broadcasts, reductions, and parallel scans. These will be used throughout our algorithms. We present algorithms with logarithmic depth and energy that is optimal for square subgrids.

\subsection{Broadcast}
Consider the problem of broadcasting a value from the processor $(0, 0)$ to all other processors in an $h\times w$ sub-grid that contains $p_{0, 0}$, where w.l.o.g. $h\geq w$ and $n=wh$. We start with a lower bound.
\begin{lem}
A broadcast on a $\Omega(h) \times \Omega(w)$ subgrid takes $\Omega(\log (h) + \log(w))$ depth and $\Omega(h+w)$ wire-depth.
\end{lem}
\begin{proof}
Consider a sequence of messages that depend on each other. For the broadcast to be correct, there must be such a sequence starting from the root and ending at every processor. Because each processor can only send a message to a constant number of other processors in each time step and we need to reach $h\times w$ processors, it requires $\Omega(\log (h \cdot w ))$ time steps to do so.
\end{proof}

\begin{figure}
	\includegraphics[width=.8\linewidth]{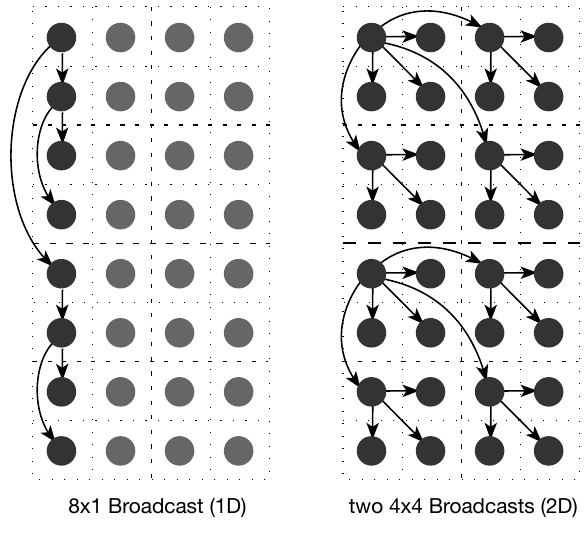} \vspace{-1em}
	\caption{Example of an 8x4 broadcast.  A 1D Broadcast sends a message along a binary tree. Then, a set of 2D Broadcasts send messages recursively along the quadrants.}\label{fig:bcast}
\end{figure}

Now, we turn to the upper bound. Let us first consider the square $w\times w$ case (2D broadcast):
We can solve the problem efficiently by subdividing the grid into quadrants and proceeding recursively on them:
Send the value to the three processors in the top-left corners of the other quadrants, then proceed recursively on each quadrant.

Next consider the $h\times 1$ case (1D broadcast). We build a binary broadcast tree, as follows. The root has a child node directly next to it and a child node at an offset $(h-1)/2$. Then, recursively build a binary tree for each child's subtree (each contains $(h-1)/2$ elements). \Cref{fig:bcast} shows an example of a 1D broadcast for $n=8$ and a 2D broadcast for $n=16$.

Now, consider the general case, where we want to broadcast on an $h\times w$ grid, where $h\geq w$. 
First, do a 1D broadcast on the first column. Then, subdivide the grid into square $w\times w$ subgrids and run a 2D broadcast on each of them.

\begin{lem}
Broadcast on an $h\times w$ subgrid (where $h\geq w$) takes $O(hw + h \log h)$ energy, $O(\log n)$ depth, and $O(w+h)$ wire-depth.
\end{lem}
\begin{proof}
The 1D broadcast takes $O(h\log h)$ energy: Let $E_1(h)$ be the energy required for the $h\times 1$ broadcast. Then, we have that:
\begin{align*}
E_1(h) \leq
\begin{cases}
	& 0  \enspace \text{if $w \leq 1$} \\
	& \frac{h}{2} + 1 + 2 E_1(w/2) \enspace \text{otherwise, }
\end{cases}
\end{align*}
which solves to $O(h \log h)$.
The following $\ceil{\frac{h}{w}}$ 2D broadcasts take $O(w^2)$ energy each: 
Let $E_2(w)$ be the energy required for the 2D broadcast on a $w\times w$ subgrid. Then, we have that:
\begin{align*}
E_2(w) \leq
\begin{cases}
	& 0  \enspace \text{if $w \leq 1$} \\
	& \frac{3w}{2} + 3 + 4 E_2(w/2) \enspace \text{otherwise, }
\end{cases}
\end{align*}
which solves to $O(w^2)$. The depth is clearly $O(\log w + \log h)=O(\log n)$. 
Finally, the wire-depths of both the 1D and 2D broadcast form geometric series and solve to $O(w+h)$.
\end{proof}

Interestingly, the energy upper bound depends on the shape of the subgrid. As long as the subgrid is not exponentially tall (that is, $h \leq e^{O(w)}$), the energy is linear in the number of processors in the subgrid. This matches the trivial energy lower bound of $w \cdot h$ up to constant factors (when $h \leq e^{O(w)}$).

\subsection{Reduce}
Given an associative and commutative operator $\circ$ and $n$ inputs $A_0, \dotsc, A_{n-1}$ stored in arbitrary order on an $h\times w$ subgrid containing the processor $p_{0, 0}$, a \emph{reduce} computes $A_0 \circ A_1 \circ \dotsc \circ A_{n-1}$ and leaves the result in the root processor $p_{0, 0}$.
To compute a reduce on a subgrid, we can use the reverse commmunication pattern as the broadcast. Hence, the result follows:
\begin{cor}
Reduce on an $h \times w$ subgrid (where $h\geq w$) takes $O(hw + h \log h)$ energy, $O(\log n)$ depth, and $O(w + h)$ wire-depth.
\end{cor}
\noindent
An \emph{All-reduce} can be implemented by a reduce followed by a broadcast. Hence, we conclude:
\begin{cor}
All-Reduce on an $h \times w$ subgrid (where $h\geq w$) takes $O(hw + h \log h)$ energy, $O(\log n)$ depth, and $O(w+h)$ wire-depth.
\end{cor}

\subsection{Parallel Scan}

Consider an array $A_0, \dotsc, A_{n-1}$ stored on a $w \times w$ grid (n=$w^2$) in some order. We want to compute the prefix sums $$A_0, \ A_0 + A_1, \ A_0 + A_1 + A_2, \dotsc , \sum_{i=0}^{n-1} A_i$$ where the i-th result $\sum_{j=0}^{i} A_j$ is stored where the $i$-th input was stored. The addition can be replaced by any associative operator, resulting in a \emph{parallel scan}. For ease of notation, we will present the algorithm for the special case of addition (w.l.o.g.).

\paragraph{Z-Order Curve}
If the array is stored in row-major order on the grid, it is challenging to implement an energy-efficient parallel scan. However, when the array is stored according to a locality preserving traversal of the grid, we can adapt a parallel up-sweep/down-sweep algorithm. 
%
%
The \emph{Z-Order curve} (sometimes called Morton space filling curve~\cite{66MortonZorder, DBLP:journals/focm/BursteddeHI19}) is one such traversal of a grid. We can define it recursively: For a $w\times w$ subgrid, traverse the four quadrants in order (visiting the top two quadrants first, left to right, then the bottom two quadrants, left to right). 

\begin{observation}
Sending a message along each edge of a Z-Order curve of a $w\times w$ subgrid takes $O(w^2)$ energy.
\end{observation}

%

We now describe the up-sweep and down-sweep of the parallel scan.
The up-sweep computes partial sums on a quadrant-level granularity. Then, the down-sweep uses these results for the final prefix sums.
Conceptually, we build a complete $4$-ary summation tree over the quadrants. The root corresponds to the whole subgrid, its children are the quadrants of the subgrid, and so on, recursively.

\paragraph{Up-sweep}
For each node in the quadrant tree, we want to compute the sum of its leaf elements. If the current subgrid contains a single processor, its element equals the value at a leaf. Otherwise, proceed recursively.
\begin{itemize}
	\item Recurse over all quadrants to obtain the sum of the element of the children.
	\item Sum those values, store the result in the $i$-th processor of the current subgrid (in Z-order), where $i$ is the distance to a leaf in the $4$-ary summation tree.
\end{itemize}
For an example of how the summation tree is mapped to the grid, see \Cref{fig:upsweep}.

\begin{figure}
    \centering
\begin{subfigure}[t]{0.5\linewidth}
	\includegraphics[width=\textwidth]{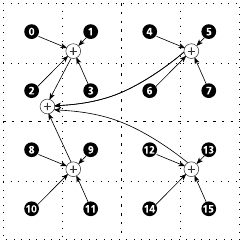} 
	\caption{The up-sweep computes partial sums along a 4-ary tree. The root of a height $i$ subtree is in the $i$-th node in Z-Order of the subtree's quadrant.}\label{fig:upsweep}
\end{subfigure}
\hfill
\begin{subfigure}[t]{0.45\linewidth}
	\includegraphics[width=\textwidth]{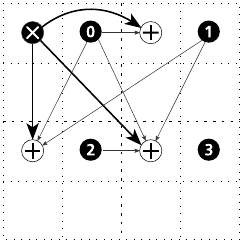} 
	\caption{The down-sweep computes prefix sums over the values from the up-sweep. It sends these prefix sums to the quadrants recursively.}
	\label{fig:down-sweep}
\end{subfigure}
\caption{The energy-optimal parallel scan operates recursively in Z-order.}
\end{figure}

\paragraph{Down-sweep} Now, we use the values from the up-sweep to compute the prefix sum. The algorithm again recursively considers the subgrid's quadrants. At each step, a value $x$ gets passed down from above. For the first invocation, $x=0$. This values $x$ is added to all values in the current quadrant to account for values that occur outside the current subgrid.

If the current subgrid has size $1$, add $x$ to the value of $A$ in the subgrid's only processor. Otherwise, proceed recursively. Consider the four quadrants in the current subgrid, $S_0, S_1, S_2, S_3$ in Z-order and their respective values $s_0, s_1, s_2, s_3$ that were computed during the up-sweep. The value that gets passed down to quadrant $S_i$ is $x+\sum_{j=0}^{i-1} s_j$. These values are passed down to the top left processor of each subgrid. \Cref{fig:down-sweep} illustrates one step of the down-sweep.
\begin{lem}
A parallel scan on an array of $n$ elements in Z-order takes $O(n)$ energy, $O(\log n)$ depth, and $O(\sqrt{n})$ wire-depth.
\end{lem}
\begin{proof}
The energy equals that of a Z-order curve up to constant factors. Note that in the up-sweep, each processor stores at most $2$ values of the summation tree. Because of the recursive construction, the wire depth forms a geometric series.
\end{proof}

\paragraph{Segmented Scan} To operate efficiently on multiple consecutively stored arrays, we consider \emph{segmented scans}. The input array is partitioned into consecutive \emph{segments}. The result of a segmented scan is the same as executing a scan on each segment. For any associative operator, we can define a segmented associative operator that has the logic of the segments built-in~\cite{ReifSynthesis}. Hence, we can use the same algorithm for a segmented scan:

\begin{cor}
A segmented scan on an array of $n$ elements in Z-order takes $O(n)$ energy, $O(\log n)$ depth, and $O(\sqrt{n})$ wire-depth.
\end{cor}

\section{Selection and Sorting}

We present a low-depth energy-optimal sorting algorithm. It is easy to prove a non-linear lower bound on the energy to permute (and hence sort) the input.
\begin{lem}\label{lem:permutation}
Permuting $h \times w$ elements on an $h \times w$ subgrid takes $\Omega(\max(w, h)^2 \cdot \min(w, h))$ energy.
\end{lem}
\begin{proof}
Assume w.l.o.g. that $h>w$. Otherwise, the situation is transposed. Consider the permutation that reverses the order in which elements appear in a row-wise layout. The elements in the first $h/3$ rows need to be sent to one of the last $h/3$ rows, which takes at least $h/3$ energy per element (of which there are $hw / 3$).
\end{proof}
The lowest permutation cost is obtained when $w=h$. Since sorting can be used to implement permutation, the sorting lower bound follows.
\begin{cor}
Sorting $n$ elements takes $\Omega(n^\frac{3}{2})$ energy.
\end{cor}
By the permutation lower bound, the matching upper bound can only be obtained when the input is contained in a $h\times w$ subgrid where $w=\Theta(h)$. Hence, we will focus on the case where $w=h$. We begin by analysing the energy of a sorting network in our model, which we show to be sub-optimal. Then, we present an energy-optimal 2D Mergesort algorithm, which has $O(\log ^ 3 n)$ depth and $O(n^{3/2})$ energy on a $\sqrt{n}\times\sqrt{n}$ grid. 
%
%

In light of the permutation lower bound, it is surprising that, as we will show, \emph{median selection} (and more generally rank selection) takes linear $O(n)$ energy.

\subsection{Sorting Networks}

\begin{figure}
    \centering
\begin{subfigure}[t]{0.38\linewidth}
	\includegraphics[ width=1\linewidth]{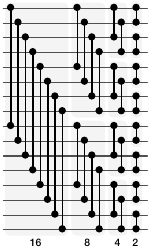} \hfill \vspace{-1em}
	\caption{1D Bitonic Merge}
	\label{fig:bitonic-merge}
\end{subfigure}
\hfill
\begin{subfigure}[t]{0.56\linewidth}
	\includegraphics[width=1\linewidth]{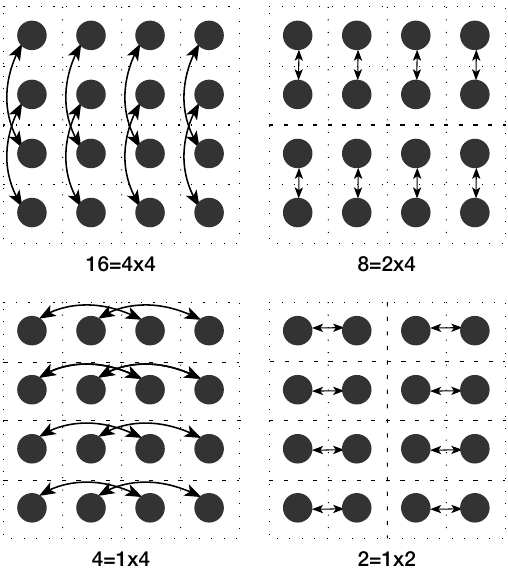} \vspace{-1em}
	\caption{2D Bitonic Merge}
	\label{fig:bitonic-merge-2d}
\end{subfigure}\vspace{-0.5em}
	\caption{Bitonic Merge network in 1D and 2D layout. In 2D, the recursion first splits into fewer columns (4x4 - 2x4 - 1x4), then fewer rows (1x2). Because the recursion does not reduce both rows and columns simultaneously, 2D Bitonic Mergesort is energy-suboptimal.}\label{fig:bitonic}
\end{figure}
Sorting networks are data-oblivious (and stable) sorting algorithms of oftentimes low depth~\cite{DBLP:conf/afips/Batcher68, DBLP:conf/stoc/AjtaiKS83, DBLP:journals/algorithmica/Paterson90, DBLP:journals/vldb/MullerTA12} . For each time step, they define a set of pairs of indices to compare (and swap if necessary). Each index into the array is thought of as a "wire". In each step, a wire can be compared with at most one other wire. A natural idea is to map a sorting network to our processor grid: each wire in the sorting network is assigned to a processor. For example, we can assign wires to processors in row-major order. Interestingly, this approach does not easily lead to energy-optimal sorting algorithms. We present the results for Bitonic Sort~\cite{DBLP:conf/afips/Batcher68}. 

The Bitonic Sort is a simple network with $O(\log ^2 n)$ depth and $O(n \log^2 n)$ comparisons. As it is defined recursively on halves of its input, it exhibits some degree of spatial locality. A Bitonic Sort makes use of a Bitonic Merge network, which can be defined recursively: For an input of size $n$, compare and swap each wire $i$ with index less than $n/2$ with wire $i+n/2$. Then, recursively merge both halves. See \Cref{fig:bitonic} for an illustration of a $4\times 4$ Bitonic Merge with wires mapped to the compute grid in row-major order. A Bitonic Sort recursively invokes itself on both halves of its input, then invokes a Bitonic Merge on the input.

We begin with the analysis of the recursive Bitonic Merge.
%
\begin{lemma}
On an $h\times w$ subgrid, Bitonic Merge takes $\Theta( h^2 w + w^2 h)$ energy, $\Theta(\log n)$ depth, and $\Theta(w+h)$ wire-depth.
\end{lemma}
\begin{proof}
We split the energy cost into two parts: (1) when there is more than one row left; (2) when there is a single row left.
When there are $h>1$ rows left, the network sends $\Theta(w\cdot h)$ messages across a distance of $h/2$. Hence, the energy $E_1(h)$ for this part is
$E_1(h) = \Theta( h^2 w) + 2 E(h/2) \ \text{if $h>1$, }$
which solves to $E_1(h)=\Theta( h^2 w)$. 
When there is a single row left of length $w$, the network sends $\Theta(w)$ messages across a distance of $\Theta(w)$. Hence, the energy $E_2(w)$ for this part is 
$E_2(w) = \Theta(w^2) + 2 E(w/2) \ \text{, }$
which is in $\Theta(w^2)$. The algorithm reaches the situation $h$ times, meaning that the total cost of this part is $O(h w^2)$. 
The wire-depth is a geometric series (first over $h$, then over $w$).
\end{proof}

Next, we describe and analyze the cost of the bitonic sorting network.  
Because the Bitonic Sort has a 1D recursive pattern that first reduces the number of rows, and only then the number of columns, it has to pay the energy of the Bitonic Merge a logarithmic number of times in one dimension, leading to the following bound:
\begin{lemma}
On an $h\times w$ subgrid, Bitonic Sort takes $\Theta( h^2 w + w^2 h \log h)$ energy, $\Theta(\log ^2 n)$ depth, and $\Theta(h + w \log h)$ wire-depth.
\end{lemma}
\begin{proof}
We again divide the energy cost into the part where there is more than one row and the part when there is a single row left. 
The energy $E_1(h)$ when there are $h>1$ rows left is
\begin{align*}
E_1(h) \leq O( h^2 w + w^2 h) + 2 E(h/2) \enspace \text{if $h>1$. }
\end{align*}
We can see that $E_1(h)=O( h^2 w + w^2 h \log h)$. 
When there is a single row left of length $w$, the energy $E_2(w)$ for this part is 
\begin{align*}
E_2(w) \leq O(w^2) + 2 E(w/2) \enspace \text{, }
\end{align*}
which solves to $O(w^2)$. This cost occurs $h$ times.

For the wire-depth, observe that as long as there are is more than one row, the cost is $O(w+h)$. Because $w$ stays the same while there is more than one row, the wire-depth is $O(w\log h + h)$.
\end{proof}
In conclusion, Bitonic Sort takes $O(n^{\frac{3}{2}}  \log n)$ energy to sort $n$ numbers, a logarithmic factor away from optimal. Note that the sub-optimality is not because of the suboptimal number of comparisons performed by the bitonic sorting network, but instead because the network eventually turns into a 1D algorithm (when the recursion becomes smaller than a single row). Moreover, Bitonic Sort is also not wire-depth optimal, as it has $\Theta(\sqrt{n}\log n)$ wire-depth. We now present an energy and wire-depth optimal algorithm. 

\subsection{Energy-Optimal Sorting}

We design a spatial energy-optimal variant of Mergesort~\cite{DBLP:journals/siamcomp/Cole88, DBLP:journals/siamcomp/Valiant75, DBLP:journals/cacm/HuangL88}. On a high level, the 2D Mergesort works similarly to traditional Mergesort. However, instead of recursing on two halves of the array, we recurse on the four quadrants of the subgrid:
\begin{itemize}
	\item Recursively sort the four quadrants of the subgrid.
	\item Merge the two top quadrants
	\item Merge the two bottom quadrants
	\item Merge the results of the two previous merges
\end{itemize}
The challenge lies in an energy-efficient implementation of the Merge subroutine, that we present in the rest of this subsection. Our merging subrouting relies on a naive sorting routine (All-Pairs Sort), which we discuss next.
 

\paragraph{All-Pairs Sort}
A simple idea for a low-depth sorting algorithm is to compare every element with every other element. This can be done by using our efficient broadcast and reduce patterns. The implementation ``explodes'' the computation onto a larger subgrid. This leads to low depth. However, because the computation grid has a larger diameter, the energy cost increases to more than quadratic. 
\begin{itemize}
	\item Subdivide a $n\times n$ subgrid into $n$ subgrids $\Gamma_i$ of size $\sqrt{n}\times\sqrt{n}$ each. Scatter the elements of $A$ such that $A_i$ is sent to the first processor of the subgrid $\Gamma_i$, for each $i$.
	\item Within each subgrid $\Gamma_i$, broadcast the element $A_i$.
	\item Copy the array $A$ to each grid $\Gamma_i$ by using the same communication pattern as for the 2D broadcast (by treating the array $A$ and the subgrids $\Gamma_i$ as units).
	\item Now, each processor compares its two elements.
	\item Each subgrid $\Gamma_i$ performs a reduction to compute the rank of element $A_i$ in the sorted sequence. Gather these ranks.
\end{itemize}

\begin{lemma}\label{lem:all-pairs-sort}
All-Pairs Sort of $n$ elements takes $O(n^{5/2})$ energy, $O(\log n)$ depth, and $O(n)$ wire-depth.
\end{lemma}
\begin{proof}
Scattering the $n$ elements over a distance of at most $O(n)$ takes $O(n^2)$ energy. Let $E(k)$ be the energy of the broadcast when $k \cdot k$ subgrids remain. Then, we have the recurrence
\begin{align*}
E(k) \leq
\begin{cases}
	& O(n)  \enspace \text{if $k \leq 1$ \ ;} \\
	& n^{3/2} k + 4 E(k/2) \enspace \text{otherwise, }
\end{cases}
\end{align*}
which is in $O(n^{3/2} k^2)$. As initially, $k=\sqrt{n}$, we get that the energy is $O(n^{5/2})$. Computing the reductions takes $O(n^2)$ energy and gathering these ranks takes $O(n^{2})$ energy.

The depth is bottlenecked by the broadcasts. Finally, the wire-depth is bottlenecked by scattering the elements in the $n \times n$ subgrid. 
\end{proof}

\paragraph{Merging Two Sorted Arrays}

\begin{figure}
	\includegraphics[width=.83\linewidth]{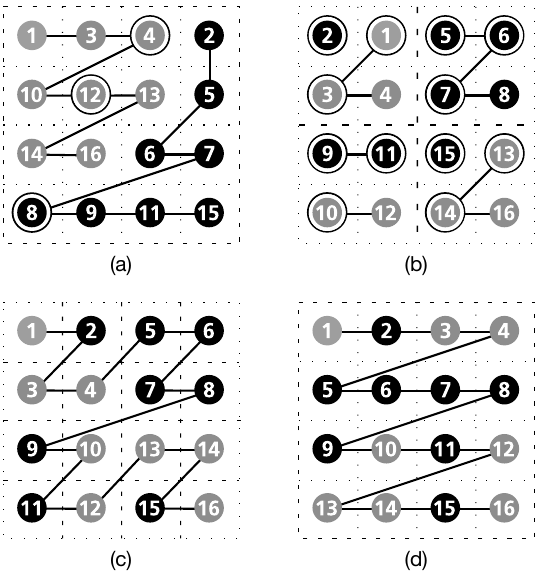} \vspace{-0.5em} 
	\caption{(a-c) The 2D merge recursively splits the two sorted arrays (colored in black and grey) by the encircled rank $n/4$, $n/2$ and $\frac{3}{4}n$ elements into quadrants. (d) Finally, it permutes the array from Z-Order into row-major order.}\label{fig:merge} \vspace{-1em} 
\end{figure}

The challenging part of our merging algorithm is an energy-efficient and low-depth merging subroutine. We cannot use classical approaches because they do not lead to balanced recursive cases (which is needed to organize them into square subgrids). 
Consider two arrays $A$ and $B$ containing $n_A$ and $n_B$ sorted numbers in row-major order. The goal is to construct an array $C$ that contains the $n=n_A+n_B$ elements of $A$ and $B$ in sorted row-major order.

%

At any point in the recursion, the algorithm operates on a square subgrid $\Gamma$. The larger of the two current subarrays $A$ and $B$ is stored in a square subgrid $\Gamma' \subseteq \Gamma$, while the other array is stored in row-major order filling up the rest of the subgrid $\Gamma$, in a ``mirrored L'' shape, as illustrated in \Cref{fig:merge}a. The idea of the algorithm is as follows.
\begin{itemize}
    \item Let $A \Vert B$ be the concatenation of $A$ and $B$. Split $A$ and $B$ by the rank $n/4$, rank $n/2$ and rank $3n/4$ elements of $A \Vert B$ into 4 subarrays each.
    \item Reorganize the resulting subarrays into the four quadrants, such that the first quadrant contains the $n/4$ smallest elements of $A \Vert B$, and so on.
    \item Recursively merge each quadrant.
    \item After finishing the recursion, the array is sorted in Z-Order. Hence, permute the array to row-major order. 
\end{itemize}
See \Cref{fig:merge} for an example.

\paragraph{Rank selection in two sorted arrays.}
To implement the merge, we still need to answer how to efficiently find the rank $k$ element of $A \Vert B$ (in particular for $k=n/4, n/2, 3n/4$). The idea is to use \emph{deterministic sampling} to quickly rank a subset of the elements and determine much smaller subarrays of $A$ and $B$ that contain the rank-$k$ element. 
\begin{itemize}
	\item Gather every $\floor{\sqrt{n}}$-th element of $A$ and $B$ into a sample $S$. Specifically, select from $A$ the elements at indices $i\floor{\sqrt{n}}$ for $i$ in the range $1, \dotsc, \floor{n_A / \floor{\sqrt{n}}}$ and similarly for $B$.
	\item Sort the sample $S$ with an All-Pairs Sort.
	\item Let $l=\lfloor \frac{k-1}{\lfloor{\sqrt{n} \rfloor}} \rfloor$.
	\item If $l=0$ let $s_l(A)$ and $s_l(B)$ refer to the first elements of $A$ and $B$, respectively.  
	\item If $l>0$, let $s_l$ be the rank $l$-th element in $S$. Find the predecessor $s_l(A)$ and $s_l(B)$ of $s_l$ among the sampled elements of $A$ and $B$, respectively (with a binary search). If no such element exists, set it to the first element in the array.
	\item In any case, let $s_r(A)$ and $s_r(B)$ be the next sampled element following $s_l$ in $A$ and $B$, respectively. If no such element exists, set it to the last element in the array.
	\item Narrow the search to the two smallest subarrays that contain the elements $s_l(A), s_r(A)$ and $s_l(B), s_r(B)$.
	\item Determine the rank $k-l \lfloor \sqrt{n} \rfloor$ element in the concatenation of the narrowed subarrays using All-Pairs Sort.
\end{itemize}

\begin{lemma}
Given two sorted arrays $A$ and $B$, Deterministic sampling determines the rank $k$ element in $O( (n_A+ n_B)^{5/4})$ energy, $O(\log n)$ depth, and $O(\sqrt{n})$ wire-depth.
\end{lemma}
\begin{proof}
Energy, depth, and wire-depth are bottlenecked by the All-Pairs Sort of the sample $S$. Because $S$ has $O(\sqrt{n_A+n_B})$ elements, the result follows by \Cref{lem:all-pairs-sort}.

Next, we prove the correctness of the algorithm. \emph{Case 1}: $l=0$. Then $k$ is at most $\floor{\sqrt{n}}$. Hence, it must be contained within the first $\sqrt{n}$ elements of one of the two subarrays. \emph{Case 2}: $l>0$. The rank of $s_l$ in $A||B$ is at most $l \floor {\sqrt{n}} \leq k-1$. Hence, by removing elements smaller than $s_l$ we do not exclude the rank $k$ element. Moreover, the rank of $s_l$ in $A||B$ is at least $k-1-\floor{\sqrt{n}}$. Since we always consider the next $\floor{\sqrt{n}}+1$ ranked elements, correctness follows.
\end{proof}

We now bound the cost of the merging subroutine.
\begin{lem}\label{lem:merging}
	Merging two arrays $A$ and $B$ with $n_A$ and $n_B$ elements located on adjacent square subgrids takes $O((n_A+n_B)^{3/2} )$ energy, $O(\log ^2 (n_A+n_B) )$ depth, and $O(\sqrt{n})$ wire-depth.
\end{lem}
\begin{proof}
Let $E(n_A, n_B)$ be the energy of merging arrays $A$ and $B$ with $n_A$ and $n_B$ elements in total. Then, the energy in each step of the recursion is $O( (n_A+ n_B)^{5/4}  )$ for determining the splitting elements and $O( (n_A+n_B)^{3/2} )$ for performing the necessary permutations. Each of the recursive calls operates on four pairs of subarrays. We get the recurrence for the energy 
\begin{align*}
E(n_A, n_B) \leq
\begin{cases}
	 (n_B)^{\frac{3}{2}}  \enspace \text{if $n_A = 0$} \\
	 (n_A)^{\frac{3}{2}}  \enspace \text{if $n_B = 0$} \\
	 O( (n_A+n_B)^{3/2 }) + \sum_{i=1}^{4} E(n_A^i, n_B^i)  \enspace \text{else, }
\end{cases}
\end{align*}
where $n_A^i+n_B^i=(n_A+n_B)/4$. Since the number of elements remains constant when summed over all nodes in the same level of the recursion, and the `diameter' term $\sqrt{n_A + n_B}$ decreases geometrically in each level of the recursion (for all recursive calls), the recurrence solves to $E(n_A, n_B)\leq O((n_A+n_B)^{3/2} )$.
\end{proof}
Note that for the case where $n_A=n_B$, the energy to merge the subarrays is $O(n^{\frac{3}{2}})$, the permutation bound.


Finally, we can bound the costs of 2D Mergesort:
\begin{thm}
2D Mergesort takes $O(n^{3/2})$ energy, $O(\log ^3 n)$ depth, and $O(\sqrt{n})$ wire depth on a square subgrid.
\end{thm}
\begin{proof}
By \Cref{lem:merging} , the energy $E(n)$ is
\begin{align*}
E(n) \leq
\begin{cases}
	& 0  \hfill \text{if $n \leq 1$;} \\
	& O(n^{\frac{3}{2}}) + 4E(n/4) \enspace \enspace \enspace  \hfill \text{else,}
\end{cases}
\end{align*}
which solves to $O(n^{\frac{3}{2}})$.
\end{proof}

\subsection{Energy-Optimal Selection}

We can determine the median of $n$ elements, and more generally the rank-$k$ element with \emph{linear energy} and poly-logarithmic depth. We can assume without loss of generality that $k\leq \lceil n/2 \rceil$, as otherwise we can reverse the order of the elements and select the rank $n-k$ element.

The idea is similar to the deterministic ranking we used for merging sorted subarrays. However, since we do not start with a partially sorted array, we have to sample randomly. We create a sample that is as large as possible, rank this sample, and use the ranking of the sample to eliminate a polynomial fraction of the input elements. Because it takes $O(\sqrt{n})$ energy to gather an element in a subgrid, the largest sample we can gather in $O(n)$ energy contains $O(\sqrt{n})$ elements. The idea of using sampling is similar to an idea by Reishuk~\cite{DBLP:journals/siamcomp/Reischuk85}. Our goal is to prove that $O(1)$ iterations suffice with high probability.

Initially, all elements are marked as \emph{active}. Elements will be gradually marked \emph{inactive} until we find the rank-$k$ element. Let $N$ be the current number of \emph{active} elements. Choose a constant $c$ with $c\geq 3$:
Repeat the following step until $N\leq c\sqrt{n} $.
\begin{itemize}
	\item Create a sample $S$ by including every active element independently with probability $c N^{-1/2}$.
	\item Gather those elements in a square subgrid, using a scan to assign each sampled element an index within the subgrid and a broadcast to communicate the size of the sample.
	\item Choose \emph{two pivots}. The first pivot is the element $s_r$ at rank $r=\min(|S|, c k N^{-1/2} + \frac{c}{2} N^{1/4}\sqrt{\ln n }) \enspace .$ 
	If $k\geq \frac{1}{2} N^{3/4} \sqrt{\ln n}$, the second pivot is the element $s_l$ at rank $l=c k N^{-1/2} - \frac{c}{2} N^{1/4}\sqrt{\ln n} \enspace .$
	Otherwise, the second pivot is the dummy element $s_l=-\infty$. 
	Sort the sample $S$ using Bitonic Sort to determine $s_l$ and $s_r$. 
	\item Broadcast $s_l$ and $s_r$ in the original subgrid.
	\item Count the number of active elements $N_{<l}$ smaller than $s_l$ and the number of active elements $N_{> r}$ larger than $s_r$ with an all-reduce. If $N_{<l} \geq k$ or $N_{>r}\geq N-k$, sort the input using 2D Mergesort and return the rank $k$ element. Otherwise, set $k=k-N_{<l}$ and continue.  
	\item For each active element $a$, inactivate it if $a < s_l$ or $a>s_r$.
	\item Count the number of remaining active elements $N$ with an all-reduce. If $k> \ceil{N/2}$, set $k=N-k$ and reverse the order of the elements (logically, that is, by henceforth flipping the result of the comparator).
\end{itemize}
Once the iteration terminates, gather the elements in a square subgrid, sort them and return the rank $k$ element.

The idea behind the energy efficiency proof is that the number of input elements of rank at most $k$ is highly concentrated around their expectation. Hence, the probability that the true rank $k$ element is between the pivot elements $s_l$ and $s_r$ is high. 
\begin{lem} \label{lem:sample:success}
The probability that $N_{<l} \geq k$ or $N_{>r}\geq N-k$ is at most $2n^{-c/6}$.
\end{lem}
\begin{proof}
Let $K$ be the random variable denoting the number of rank at most $k$ elements of the input that are sampled. Let $\delta=\frac{cN^{1/4} \sqrt{\ln n}}{2\E[K]}$. 
We first consider the case where $k\geq \frac{1}{2} N^{3/4} \sqrt{\ln n}$ and there are two non-trivial pivots. Observe that $N_{<l} \geq k$ occurs when $K > l$ and $N_{>r} \geq N-k$ occurs when $K<r$. Note that $\E[K]=ckN^{-1/2}$. Hence, it remains to bound the probability that $K$ deviates from its expectation by more than $\frac{c}{2}N^{1/4}\sqrt{ \ln n }$. Note that $0<\delta\leq 1$. By a Chernoff bound~\cite{Doerr2020}, we get
\begin{align*}
P\left[|K-\E[K]| \geq \frac{c}{2} N^{1/4}\sqrt{\ln n}\right] &= P\left[|K-\E[K]| \geq \delta \E[K] \right] \\
&\leq 2e^{-\delta^2 \E[k]/3}\\
&\leq 2e^{-\frac{c^2 N^{1/2}\ln n}{12 \E[K]}}\\
&\leq 2n^{-c/6} \enspace . 
\end{align*}

For the case where $k< \frac{1}{2}N^{3/4} \sqrt{\ln n}$, we have that $N_{<l}=0$. Thus, we only need to bound $P[K \geq (1+\delta)\E[k]]$. Note that $\delta > 1$. We can conclude by another Chernoff bound~\cite{Doerr2020}:
\begin{align*}
P\left[K \geq (1+\delta) \E[k] \right] 
\ \leq \ e^{\delta \E[k]/3}
\ < \ e^{-2N^{1/4}} \enspace .
\end{align*}
\end{proof}

Next, we bound the size of the number of active elements after one iteration. The idea is that it is unlikely that more than the expected number of elements are between $s_l$ and $s_r$. 
\begin{lem}\label{lem:sample:runtime}
Let $N_{t}$ be the number of active elements after the $t$-th iteration, $N_0=n$. Given $N_t=n_t$ and any constant $0<\epsilon<1$, with probability at least $1-e^{- c \epsilon n_t^{1/4} \sqrt{\ln n} /4 }$, we have that $N_{t+1} \leq (1+\epsilon) n_t^{3/4} \sqrt{\ln n}$.
\end{lem}
\begin{proof}

We define a binomially distributed random variable $X$ to bound the probability, as follows. Consider the rank of $s_l$ within the array after the $t$-th iteration. Now consider the next $(1+\epsilon) n_t^{3/4} \sqrt{\ln n}$ subsequently ranked elements (after the $t$ iteration) in order. If an element is sampled in the $(t+1)$-th iteration, it is counted as a success. Recall that this occurs with probability $cn_t^{-1/2}$. The event that $N_{t+1} > (1+\epsilon) n_t^{3/4} \sqrt{\ln n}$ occurs exactly when $X\leq c n_t^{1/4} \sqrt{\ln n}$.
Note that $\E[X]= (1+\epsilon) c n_t^{1/4} \sqrt{\ln n}$. We bound the tail probability of $X$ by a Chernoff bound (for $\delta=\frac{\epsilon}{1+\epsilon}$): 
\begin{align*}
P[X \leq c n_t^{1/4} \sqrt{ \ln n}] &= P[ X \leq (1-\delta)  \E[X] ] \\
& \leq e^{-\delta^2 \E[X] / 2} \\
& = e^{- \frac{(\frac{\epsilon}{1+\epsilon}) c n_t^{1/4} \sqrt{\ln n} ) }{2}} \\
& \leq e^{-\frac{c\epsilon}{4} n_t^{1/4} \sqrt{\ln n}  \enspace .} 
\end{align*}
\end{proof}

\begin{thm}
Rank Selection takes $O(n)$ energy, $O(\log^2 n)$ depth, and $O(\sqrt{n})$ wire-depth w.h.p. in $n$ (and also in expectation).
\end{thm}\label{lem:select}
\begin{proof}
It takes $O(n)$ energy to send the $O(\sqrt{n})$ sampled elements across the $O(\sqrt{n})$ diameter compute grid. Sorting the sample takes $O(n^{3/4} \log n)=o(n)$ energy. The remaining operations take $O(n)$ energy using our communication primitives. Hence, each iteration takes $O(n)$ energy. The depth is bottlenecked by the Bitonic Sort, which takes $O(\log ^2 n)$ depth. The wire-depth is $O(\sqrt{n})$ in each iteration. For all $N_t$ larger than a constant, by \Cref{lem:sample:runtime} we have that $N_{t+1}\leq N_t^{4/5}$ with high probability in $n$. Hence, the algorithm performs a constant number of iterations. By \Cref{lem:sample:success}, each of those iterations resorts to sorting the whole input with probability at most $2n^{-{c/6}} \leq 2n^{-1/2}$, which implies the expectation bounds. By choosing an appropriate constant $c$, the probability of success can be boosted to $1-n^{-d}$ for any constant $d$.
\end{proof}

\section{Matrix Multiplication}

For matrix multiplication algorithms, the size of the subgrid used for computation plays a crucial role in determining the energy cost. Generally, using a larger subgrid (by replicating parts of the matrices) allows for lower depth but leads to a larger energy cost. We focus on square matrix multiplication, for which we show tight energy upper and lower bounds. 

Interestingly, we can obtain a tight lower bound on the energy to multiply two square matrices. The idea is to use a permutation matrix as one of the matrices.
\begin{lem}
Multiplying two $n\times n$ matrices stored on a $n\times n$ subgrid in row-major order takes $\Omega(n^3)$ energy.
\end{lem}
\begin{proof}
Consider an arbitrary matrix $\matr{B}$ and a permutation matrix $\matr{A}$ that reverses the order of the rows in $\matr{B}$. Then, $\matr{A}\matr{B}$ results in the same permutation of $\matr{B}$ as in the proof of \Cref{lem:permutation}, which takes at least $\frac{n^3}{9}$ energy.
\end{proof}

\begin{figure}
	\includegraphics[width=.49\linewidth]{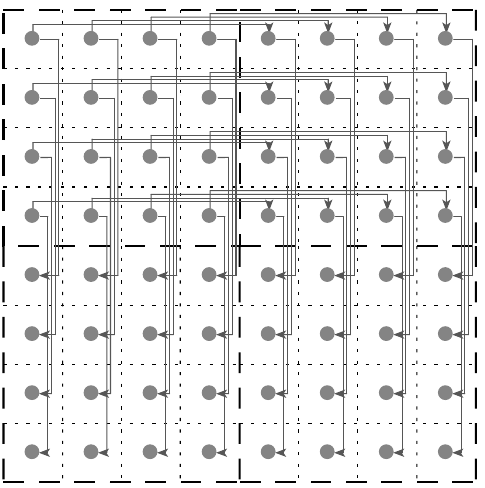} \hfill
	\includegraphics[width=.49\linewidth]{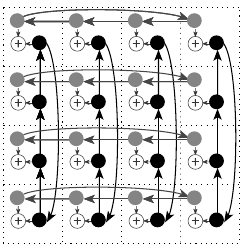} 
	\includegraphics[width=1.\linewidth]{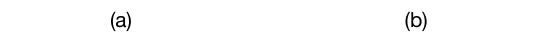} 
	\vspace{-2.3em}
	\caption{
	The matrices $\matr{A}$ and $\matr{B}$ are grey and black, respectively. (a) In the naive matrix multiply, the matrix $\matr{A}$ is copied $n-1$ times. For $n=4$, this leads to a $16\times 16$ subgrid. (b) One shifting step of Cannon's algorithm reuses the $4 \times 4$ subgrid. } \label{fig:cannon}
\end{figure}

\subsection{Schedules for Cubic Matrix Multiply}
The naive algorithm of decomposing the matrix-matrix multiplication into matrix-vector products is energy-inefficient: It involves broadcasting $\matr{A}$ and the columns of $\matr{B}$ to separate parts of the grid, which takes $\Omega(n^{9/2})$ energy for sending $\matr{A}$. The reason is that the algorithm uses a large subgrid (see \Cref{fig:cannon}a), which increases the distances needed to communicate.

%
\paragraph{Cannon's algorithm}
A 2-D matrix multiplication scheme by Cannon~\cite{69CannonCellular} is energy-optimal in our model, but has $\Theta(n)$ depth. It gradually shifts the rows and columns of the matrices $\matr{A}$ and $\matr{B}$ and accumulates local products into the result matrix $\matr{C}$, as follows.

Initially, the matrices $A$ and $B$ are distributed in row-major on the $n \times n$ subgrid, such that processor $(i, j)$ stores $A[i, j]$ and $B[i, j]$.
Shift the $i$-th row of $A$ in a circular way by $i$ to the left. Shift the $j$-th column of $B$ it in a circular way  by $j$ towards the top. In particular, the elements that were on the diagonals of $A$ and $B$ end up in the first column or row, respectively.
Each processor $(i, j)$ initializes $C[i, j]$ to zero.
Then, for $n$ iterations: Each processor $(i, j)$ adds the product of the currently colocated elements of $A$ and $B$ to $C[i, j]$. Circularly shift the rows of $A$ towards the left and columns of $B$ towards the top by one. See \Cref{fig:cannon}b for an illustration.
\begin{lem}
Cannon's algorithm takes $O(n^3)$ energy, $\Theta(n)$ depth, and $\Theta(n)$ wire-depth to multiply two $n\times n$ matrices.
\end{lem}
\begin{proof}
The initial permutation costs $O(n^3)$ energy. Then, each of the $n-1$ iterations costs $O(n^2)$ energy.
\end{proof}
One of the reasons that Cannon's algorithm is space-efficient is that it reuses its $O(n^2)$ sized computation subgrid, avoiding energy-intensive communication over long distances. 

\subsection{Schedules for Strassen's Algorithm}

To avoid the problem of requiring very large compute grids for accumulating partial results, we turn to fast matrix multiplication algorithms based on Strassen's block-recursive algorithm~\cite{10.1007/BF02165411}. When unrolling its recursion multiple times, an opportunity for a sublinear-depth algorithm arises that uses less extra space. Hence, the compute grid stays smaller compared to classic matrix multiplication, which in turn reduces the energy compared to the classic 2.5D algorithm.

\paragraph{Sublinear Depth Matrix Multiply}

Next, we present an algorithm that is almost energy-optimal, while achieving a sublinear depth. It is a schedule for  Strassen's algorithm~\cite{10.1007/BF02165411} and is a factor $O(\sqrt{\log n})$ away from energy optimal. The suboptimality is from needing $\Theta(\log n)$ memory per matrix entry to hold intermediate results.

%
The \emph{Space-Sharing Strassen's Matrix Multiply (S3MM)}, works as follows. In the original Strassen's algorithm, each recursive step produces $7$ calls to block matrix multiplications of matrices with side-lengths $n/2$. By unrolling the recursion twice, this gives us $4^3=64$ submatrices with side-lengths $n/8$ and $7^3=343$ recursive calls. We map the recursive calls in chunks of $64$ calls evenly onto the compute grid and execute those $64$ calls in parallel. The $6$ chunks run sequentially one after the other (the last one has only $23$ recursive calls). After the last chunk finishes, we accumulate the results. Saving the intermediate results (of each chunk) requires $O(1)$ memory per processor per depth in the recursion. We can get this amount of memory on an $n\sqrt{\log n} \times n \sqrt{\log n}$ subgrid.

\begin{lem}
The S3MM algorithm takes $O\left( n^3 \sqrt{\log n} \right)$ energy, \linebreak $O\left(n^{\log_8(6)}\right))$ depth, and $O(n\sqrt{\log n})$ wire-depth to multiply two $n \times n$ matrices.
\end{lem}
\begin{proof}
The depth $D(n)$ satisfies
\begin{align*}
D(n) \leq
\begin{cases}
\ O(1) &\text{if $n \leq 1$} \\
\ O(1) + 6 \ D(n/8) & \text{else, }
\end{cases}
\end{align*}
which solves to $O(n^{\log_8(6)})$. 
The wire-depth $W(n)$ is given by
\begin{align*}
W(n) \leq
\begin{cases}
\ O(1) &\text{if $n \leq 1$} \\
\ O(n \sqrt{\log n}) + 6 \ D(n/8) & \text{else, }
\end{cases}
\end{align*}
which solves to $O(n\sqrt{\log n})$. 
In each recursive step, the algorithm moves a constant number of $O(n^2)$ sized submatrices inside a subgrid of size $n\sqrt{\log n} \ \times \ n \sqrt{\log n}$. Hence, the energy $E(n)$ satisfies
\begin{align*}
E(n) \leq
\begin{cases}
	\ O(1)  \enspace &\text{if $n=1$} \\
	\ O\left( n^3 \sqrt{\log n}  \right) + 343 \ E(n/64) & \text{else, }
\end{cases}
\end{align*}
which solves to $O\left( n^3 \sqrt{\log n} \right)$
\end{proof}
Note that the depth is approximately in $O(n^{0.861})$. If we unrolled the recursion less than twice, the depth would be $\Omega(n)$ and there would be no improvement over Cannon's algorithm. By unrolling the recursion $3$ times, we could get a slightly better depth of $n^{\log_{16}(10)}$ at the cost of worse constant factors in the energy. 

\paragraph{Energy-Depth Tradeoff}

Instead of traversing the recursion tree in a depth first manner, we can traverse it in a breadth-first manner. This requires a significantly larger subgrid, which increases the energy costs. To obtain an energy-depth tradeoff, we stop traversing the recursion in a breadth-first manner and switch to S3MM once the matrix side-lengths are at most $k$. It is worth noting that this also constitutes a memory-depth tradeoff: The larger a subgrid we use for the computation, the lower the depth but the higher the energy. This approach is related to the DFS/BFS paradigm for communication efficient matrix multiplication~\cite{DBLP:conf/ipps/DemmelEFKLSS13}. 

The \emph{S3MM BFS/DFS} algorithm needs $O(n^{\log_2 7} k^{2-\log_2 7}\log k )$ processors for multiplying two $n\times n$ matrices. The subgrid has side-lengths $O(n^{\log_4 7} k^{1-\log_4 7} \sqrt{\log k})$. Each call of Strassen's algorithm creates $7$ recursive calls to matrices of size $n/2 \times n/2$. By unrolling them once, we get $49$ calls to multiply two matrices of size $n/4 \times n/4$. Arrange these block products in a regular $7 \times 7$ subdivision of the current subgrid. Then, send the submatrices to the correct subgrid and recurse on each of of the subgrids. Finally, accumulate the resulting submatrices in the top left corner of the subgrid.
\begin{thm}
S3MM BFS/DFS takes $O\left( n^3 (\frac{n}{k})^{\log_4 (7) -1} \sqrt{\log k} \right)$ energy, $O(k^{\log_8(6)} + \log n)$ depth, and $O(n (\frac{n}{k})^{\log_4 (7) -1} \sqrt{\log k})$ wire-depth to multiply two $n \times n$ matrices, for any choice of $1\leq k\leq n$. 
\end{thm}
\begin{proof}
In one step, each of the $O(1)$ matrices of $n^2$ elements have to travel a distance of $n(\frac{n}{k})^{\log_4 (7) -1} \sqrt{
\log k}$. Hence, the energy $E(n)$ is bounded by
\begin{align*}
E(n) \leq
\begin{cases}
	& O(n^3 \sqrt{\log n})  \enspace \text{if $n \leq k$} \\
	& O\left( n^3 (\frac{n}{k})^{\log_4 (7)-1} \sqrt{\log k}  \right) + 49 \ E(n/4)  \enspace \text{else, }
\end{cases}
\end{align*}
which solves to $O\left( n^3 (\frac{n}{k})^{\log_4 (7)-1} \sqrt{\log k}  \right)$ for the given range of $k$.

\end{proof}
%
%

The algorithm needs approximately $O(n^{3.41} / \ k^{0.80})$ energy, depth $O(k^{0.861} + \log n)$, and wire-depth $O(n^{1.41} / k^{0.80})$.
It remains an open question to get an energy-optimal algorithm with sublinear depth.

\section{PRAM Simulation}

To quickly obtain (sub-optimal) upper bounds for a problem, it can be convenient to simulate an existing PRAM algorithm~\cite{ReifSynthesis} in our model. For exclusive reads and writes, this can be done directly by dedicating a subgrid of processors to simulate the shared memory. The PRAM processors are simulated by another subgrid, whose processors load and store data into the simulated shared memory. For concurrent reads and writes, we can use our energy-optimal sorting algorithm and parallel scans to manage the concurrency.


Let us start with the Exclusive-Read Exclusive Write (EREW) PRAM simulation. In each synchronous time step, an EREW PRAM processor can read $O(1)$ word-sized memory cells, perform $O(1)$ computation, and write to $O(1)$ memory cells. No two processors can write or read the same memory cell in the same time step. By simulating the shared memory as a square subgrid of processors and placing the PRAM processors into a square subgrid (next to the memory), we obtain a simulation results for EREW PRAM:
\begin{lemma}
Consider an algorithm $A$ on an EREW PRAM that uses $m$ memory cells and runs in $T_p$ time steps on $p$ processors. Simulating algorithm $A$ takes $O( p (\sqrt{p}+\sqrt{m}) T_p)$ energy and $O(T_p)$ depth.
\end{lemma}

We can generalize the simulation to handle concurrent reads and writes, by exploiting our energy-optimal sorting algorithm to resolve concurrent reads and writes. This comes at the cost of an increased depth.
\begin{lemma}
Consider an algorithm $A$ on an CRCW PRAM that uses $m$ memory cells and runs in $T_p$ time steps on $p$ processors. Simulating algorithm $A$ takes $O( p (\sqrt{p}+\sqrt{m}) T_p)$ energy and $O(T_p\log^3 p)$ depth.
\end{lemma}
\begin{proof}

Index the PRAM processors and memory cells with one-dimensional indexes.
Organize the PRAM processors on a $\sqrt{p}\times\sqrt{p}$ subgrid indexed in Z-order and the PRAM memory cells on a $\sqrt{m}\times \sqrt{m}$ subgrid indexed in row major order (next to it). 
We show how to simulate a sub-step where each processor reads at most one value from the simulated global memory and writes at most one value from the simulated global memory. To simulate a PRAM step, execute $O(1)$ such sub-steps.

Let us begin with the read step.
If processor with index $i$ wants to read a value at cell $k$, it creates a tuple $(i, k)$. Then, these tuples are sorted according to the last component. Each processor $i>0$ inspects its tuple $(i, k')$ and the tuple $(i-1, k'')$ of processor $i-1$. If $k'\neq k''$ or $i=0$, then this processors reads the value at cell $k$ by sending a message there and waiting for the value. The processors perform a segmented broadcast on the received values (with segments determined by the processors that read the same cell). Finally, we need to send the results back to the processors that initiated the read. A processor $j$ that got tuple $(i, k)$ from the first sorting step and received $v$ from the segmented broadcast creates a tuple $(i, v)$. Sort these tuples by first component, interpreted as a location in the Z-order curve of the processors (convert index $i$ to a 2D location on the grid). Now, each processor has read a value from global memory. 

The write step is similar. 
If processor with index $i$ wants to write a value $v$ to cell $k$, it creates a tuple $(v, i, k)$. Then, these tuples are sorted according to the last component. Each processor $i>0$ inspects its tuple $(v', i, k')$ and the tuple $(v'', i-1, k'')$ of processor $i-1$. If $k'\neq k''$ or $i=0$, then this processors sends the value $v'$ to memory cell $k'$.

Each read/write step takes $O(p\sqrt{p})$ energy to sort the tuples, and $O(p)$ energy for the segmented broadcast. Since there are at most $2p$ accesses to the simulated shared memory, each taking $O(\sqrt{p}+\sqrt{m})$ energy, the total energy of one step is $O(p \sqrt{p} + p \sqrt{m})$. Summing over the $O(T_p)$ steps, gives the results. The depth is bottlenecked by the depth of the sorting, which is $O(\log ^ 3 n)$ for each of the $T_p$ sequential steps. 
\end{proof}

\section{Coarse-grained and Hierarchical Architectures}

We defined our model with constant memory and send/receive operations per processor. Because our model inherently rewards locality of computation, bounds in terms of constant memory translate into bounds for larger local memories and even to arbitrarily nested hierarchical architectures.

\subsection{Larger Memories}

We define an $S$-fat spatial computer as a spatial computer where each processor has $S$ local memory and queue sizes and can send and receive up to $S$ messages per time-step.


\begin{lem}\label{lem:s-fat-upper}
Consider an algorithm $A$ that takes $E$ energy, $D$ depth, and $D_w$ wire-depth on a $c$-fat spatial computer, for some constant $c$. On an $S$-fat spatial computer, it takes $O(E/\sqrt{S})$ energy, $D$ depth, and $O(D_w/\sqrt{S})$ wire-depth to run $A$.
\end{lem}
\begin{proof}
Each $S$-fat processor simulates the execution of $A$ on a subgrid of $\sqrt{S/c} \times \sqrt{S/c}$ of the $c$-fat processor. In each time-step, such an $S$-fat processor sends and receives up to $S$ messages, $c$ for each of its $S/c$ simulated $c$-fat processors. 

Consider some message $m$ sent during the execution of $A$ on the $c$-fat processor. If its sender and receiver are simulated on the same $S$-fat processor, no energy is spent to send $m$. Otherwise, the distance between the sender and receiver is reduced by a factor at least $\sqrt{S/c}$. Hence, the overall energy and wire-depth is reduced by a factor $\Theta(\sqrt{S})$. The depth does not increase from the simulation.
\end{proof}

The permutation lower bound and by extension the sorting and matrix multiplication lower bounds hold analogously: 
\begin{lem}\label{lem:s-fat-permutation}
On an $S$-fat processor, permuting $h \times w$ elements on an $h \times w$ subgrid takes $\Omega\left(\frac{ \max(w, h)^2  \min(w, h)}{\sqrt{S}}\right)$ energy.
\end{lem}
\Cref{lem:s-fat-upper} and \Cref{lem:s-fat-permutation} mean that for $S$-fat processors, we can derive matching energy upper and lower bounds for scan, sorting, rank selection, and Cannon's algorithm.

\subsection{Hierarchical Spatial Computer}\label{sec:hierarchy}

Next, we consider the setting of a hierarchical many-core. We show that our bounds imply congestion bounds in the hierarchical setting, regardless of the number of levels in the hierarchy.
 The idea is to hierarchically group processors into larger tiles. 
Formally, in a \emph{hierarchical spatial computer with $k$ levels}, the subgrid is mapped onto an infinite tree, as follows. The root node at level $0$ of the hierarchy consists of the entire subgrid. Every node in level $i>0$ ($i<k$) contains a $S_i\times S_i$ subgrid $G'$ and is connected to the unique node in level $i-1$ whose subgrid $G''$ contains $G'$. 

The cost at level $i$ is therefore the cost of an $S_i$-fat spatial processor, i.e., the communication at higher levels in the hierarchy are treated as local. Then, we get the bounds for each level:

\begin{corollary}\label{lem:many-core-bandwidth}
Consider an algorithm $A$ that takes $E$ energy, $D$ depth, and $D_w$ wire-depth. In level $i$ of a hierarchical spatial computer, the algorithm $A$ takes $O(E/\sqrt{S_i})$ energy, $D$ depth, and $O(D_w/\sqrt{S_i})$ wire-depth.
\end{corollary}
\begin{proof}
The proof of \Cref{lem:s-fat-upper} applies to all levels in the hierarchy simultaneously.
\end{proof}
We can bound the number of messages that need to travel large distances \emph{at all levels in the hierarchy, simultaneously}:

\begin{corollary}\label{lem:many-core-congestion}
Consider an algorithm $A$ that takes $E$ energy. In level $i$ of a hierarchical spatial computer, algorithm $A$ sends at most $\frac{E}{k\sqrt{S_i}}$ messages to processors that have distance $k$.
\end{corollary}
\begin{proof}
Follows from \Cref{lem:many-core-bandwidth} and the observation that at most a $1/k$ fraction of messages can travel $k$ times farther than the average.
\end{proof}

\section{Conclusion}
Cutting-edge parallel architectures require us to reason about the spatial aspects of communication to obtain the best performance. 
We present a novel model of computation suitable for such architectures, where communication costs are spatially dependent. As our model abstracts unimportant details of the architecture, we were able to design energy-optimal and low-depth algorithms for a wide range of problems. For some algorithms, polynomial energy improvements are possible compared to a PRAM simulation approach. 
The algorithms for sorting and rank-selection are relatively practical and could enable more efficient statistical evaluations and sparse computations. 
We conclude with a list of open problems, in increasing order of estimated difficulty:
\begin{itemize}
	\item Sorting with optimal energy and $O(\log^2 n )$ depth.
	\item Energy-optimal rectangular matrix multiply with depth $O(n)$.
	\item Energy-optimal $n\times n$ matrix-matrix multiply with depth $o(n)$.
	\item Energy-optimal rectangular matrix multiply with depth $o(n)$.
\end{itemize}

\bibliographystyle{plain}
\bibliography{gridbib}

\begin{thebibliography}{10}

\bibitem{DBLP:conf/spaa/AdlerDJKR98}
Micah Adler, Wolfgang Dittrich, Ben H.~H. Juurlink, Miroslaw Kutylowski, and
  Ingo Rieping.
\newblock Communication-optimal parallel minimum spanning tree algorithms
  (extended abstract).
\newblock In {\em Proceedings of the Tenth Annual {ACM} Symposium on Parallel
  Algorithms and Architectures, {SPAA} '98, Puerto Vallarta, Mexico, June 28 -
  July 2, 1998}, pages 27--36, 1998.

\bibitem{DBLP:conf/stoc/AjtaiKS83}
Mikl{\'{o}}s Ajtai, J{\'{a}}nos Koml{\'{o}}s, and Endre Szemer{\'{e}}di.
\newblock An o(n log n) sorting network.
\newblock In David~S. Johnson, Ronald Fagin, Michael~L. Fredman, David Harel,
  Richard~M. Karp, Nancy~A. Lynch, Christos~H. Papadimitriou, Ronald~L. Rivest,
  Walter~L. Ruzzo, and Joel~I. Seiferas, editors, {\em Proceedings of the 15th
  Annual {ACM} Symposium on Theory of Computing, 25-27 April, 1983, Boston,
  Massachusetts, {USA}}, pages 1--9. {ACM}, 1983.

\bibitem{DBLP:conf/ipps/ArgeGS10}
Lars Arge, Michael~T. Goodrich, and Nodari Sitchinava.
\newblock Parallel external memory graph algorithms.
\newblock In {\em 24th {IEEE} International Symposium on Parallel and
  Distributed Processing, {IPDPS} 2010, Atlanta, Georgia, USA, 19-23 April 2010
  - Conference Proceedings}, pages 1--11, 2010.

\bibitem{DBLP:conf/afips/Batcher68}
Kenneth~E. Batcher.
\newblock Sorting networks and their applications.
\newblock In {\em American Federation of Information Processing Societies:
  {AFIPS} Conference Proceedings: 1968 Spring Joint Computer Conference,
  Atlantic City, NJ, USA, 30 April - 2 May 1968}, volume~32 of {\em {AFIPS}
  Conference Proceedings}, pages 307--314. Thomson Book Company, Washington
  {D.C.}, 1968.

\bibitem{ben2019modular}
Tal Ben-Nun, Maciej Besta, Simon Huber, Alexandros~Nikolaos Ziogas, Daniel
  Peter, and Torsten Hoefler.
\newblock A modular benchmarking infrastructure for high-performance and
  reproducible deep learning.
\newblock In {\em 2019 IEEE International Parallel and Distributed Processing
  Symposium (IPDPS)}, pages 66--77. IEEE, 2019.

\bibitem{ben2019demystifying}
Tal Ben-Nun and Torsten Hoefler.
\newblock Demystifying parallel and distributed deep learning: An in-depth
  concurrency analysis.
\newblock {\em ACM Computing Surveys (CSUR)}, 52(4):1--43, 2019.

\bibitem{besta2018slim}
Maciej Besta, Syed~Minhaj Hassan, Sudhakar Yalamanchili, Rachata
  Ausavarungnirun, Onur Mutlu, and Torsten Hoefler.
\newblock Slim noc: A low-diameter on-chip network topology for high energy
  efficiency and scalability.
\newblock {\em ACM SIGPLAN Notices}, 53(2):43--55, 2018.

\bibitem{besta2022parallel}
Maciej Besta and Torsten Hoefler.
\newblock Parallel and distributed graph neural networks: An in-depth
  concurrency analysis.
\newblock {\em arXiv preprint arXiv:2205.09702}, 2022.

\bibitem{besta2019graph}
Maciej Besta, Dimitri Stanojevic, Johannes De~Fine Licht, Tal Ben-Nun, and
  Torsten Hoefler.
\newblock Graph processing on fpgas: Taxonomy, survey, challenges.
\newblock {\em arXiv preprint arXiv:1903.06697}, 2019.

\bibitem{DBLP:journals/csur/BjerregaardM06}
Tobias Bjerregaard and Shankar Mahadevan.
\newblock A survey of research and practices of network-on-chip.
\newblock {\em {ACM} Comput. Surv.}, 38(1):1, 2006.

\bibitem{DBLP:journals/focm/BursteddeHI19}
Carsten Burstedde, Johannes Holke, and Tobin Isaac.
\newblock On the number of face-connected components of morton-type
  space-filling curves.
\newblock {\em Found. Comput. Math.}, 19(4):843--868, 2019.

\bibitem{69CannonCellular}
Lynn~Elliot Cannon.
\newblock {\em A Cellular Computer to Implement the Kalman Filter Algorithm}.
\newblock PhD thesis, USA, 1969.
\newblock AAI7010025.

\bibitem{DBLP:conf/date/CavalcanteRPB21}
Matheus~A. Cavalcante, Samuel Riedel, Antonio Pullini, and Luca Benini.
\newblock Mempool: {A} shared-l1 memory many-core cluster with a low-latency
  interconnect.
\newblock In {\em Design, Automation {\&} Test in Europe Conference {\&}
  Exhibition, {DATE} 2021, Grenoble, France, February 1-5, 2021}, pages
  701--706. {IEEE}, 2021.

\bibitem{DBLP:journals/jacm/ChazelleM85}
Bernard Chazelle and Louis Monier.
\newblock A model of computation for {VLSI} with related complexity results.
\newblock {\em J. {ACM}}, 32(3):573--588, 1985.

\bibitem{DBLP:conf/asap/ChinSRZKHA17}
S.~Alexander Chin, Noriaki Sakamoto, Allan Rui, Jim Zhao, Jin~Hee Kim, Yuko
  Hara{-}Azumi, and Jason~Helge Anderson.
\newblock {CGRA-ME:} {A} unified framework for {CGRA} modelling and
  exploration.
\newblock In {\em 28th {IEEE} International Conference on Application-specific
  Systems, Architectures and Processors, {ASAP} 2017, Seattle, WA, USA, July
  10-12, 2017}, pages 184--189, 2017.

\bibitem{DBLP:journals/siamcomp/Cole88}
Richard Cole.
\newblock Parallel merge sort.
\newblock {\em {SIAM} J. Comput.}, 17(4):770--785, 1988.

\bibitem{DBLP:journals/cacm/CullerKPSSSSE96}
David~E. Culler, Richard~M. Karp, David~A. Patterson, Abhijit Sahay, Eunice~E.
  Santos, Klaus~E. Schauser, Ramesh Subramonian, and Thorsten von Eicken.
\newblock Logp: {A} practical model of parallel computation.
\newblock {\em Commun. {ACM}}, 39(11):78--85, 1996.

\bibitem{DBLP:conf/dac/DallyT01}
William~J. Dally and Brian Towles.
\newblock Route packets, not wires: On-chip interconnection networks.
\newblock In {\em Proceedings of the 38th Design Automation Conference, {DAC}
  2001, Las Vegas, NV, USA, June 18-22, 2001}, pages 684--689, 2001.

\bibitem{DBLP:conf/ipps/DemmelEFKLSS13}
James Demmel, David Eliahu, Armando Fox, Shoaib Kamil, Benjamin Lipshitz, Oded
  Schwartz, and Omer Spillinger.
\newblock Communication-optimal parallel recursive rectangular matrix
  multiplication.
\newblock In {\em 27th {IEEE} International Symposium on Parallel and
  Distributed Processing, {IPDPS} 2013, Cambridge, MA, USA, May 20-24, 2013},
  pages 261--272, 2013.

\bibitem{Doerr2020}
Benjamin Doerr.
\newblock {\em Probabilistic Tools for the Analysis of Randomized Optimization
  Heuristics}, pages 1--87.
\newblock Springer International Publishing, Cham, 2020.

\bibitem{DBLP:journals/cse/EmaniVAPSFJLNSK21}
Murali Emani, Venkatram Vishwanath, Corey Adams, Michael~E. Papka, Rick
  Stevens, Laura Florescu, Sumti Jairath, William Liu, Tejas Nama, Arvind
  Sujeeth, Volodymyr~V. Kindratenko, and Anne~C. Elster.
\newblock Accelerating scientific applications with sambanova reconfigurable
  dataflow architecture.
\newblock {\em Comput. Sci. Eng.}, 23(2):114--119, 2021.

\bibitem{DBLP:conf/cvpr/FarabetMCACL11}
Cl{\'{e}}ment Farabet, Berin Martini, B.~Corda, Polina Akselrod, Eugenio
  Culurciello, and Yann LeCun.
\newblock Neuflow: {A} runtime reconfigurable dataflow processor for vision.
\newblock In {\em {IEEE} Conference on Computer Vision and Pattern Recognition,
  {CVPR} Workshops 2011, Colorado Springs, CO, USA, 20-25 June, 2011}, pages
  109--116. {IEEE} Computer Society, 2011.

\bibitem{DBLP:conf/fpga/GaideGRB19}
Brian Gaide, Dinesh Gaitonde, Chirag Ravishankar, and Trevor Bauer.
\newblock Xilinx adaptive compute acceleration platform:
  Versal\({}^{\mbox{tm}}\) architecture.
\newblock In Kia Bazargan and Stephen Neuendorffer, editors, {\em Proceedings
  of the 2019 {ACM/SIGDA} International Symposium on Field-Programmable Gate
  Arrays, {FPGA} 2019, Seaside, CA, USA, February 24-26, 2019}, pages 84--93.
  {ACM}, 2019.

\bibitem{DBLP:conf/spaa/GerbessiotisS96}
Alexandros~V. Gerbessiotis and Constantinos~J. Siniolakis.
\newblock Deterministic sorting and randomized median finding on the {BSP}
  model.
\newblock In {\em Proceedings of the 8th Annual {ACM} Symposium on Parallel
  Algorithms and Architectures, {SPAA} '96, Padua, Italy, June 24-26, 1996},
  pages 223--232, 1996.

\bibitem{DBLP:conf/ipps/GerbessiotisS97}
Alexandros~V. Gerbessiotis and Constantinos~J. Siniolakis.
\newblock A randomized sorting algorithm on the {BSP} model.
\newblock In {\em 11th International Parallel Processing Symposium {(IPPS}
  '97), 1-5 April 1997, Geneva, Switzerland, Proceedings}, pages 293--297,
  1997.

\bibitem{DBLP:conf/dimacs/GoddardKP94}
Steve Goddard, Subodh Kumar, and Jan~F. Prins.
\newblock Connected components algorithms for mesh-connected parallel
  computers.
\newblock In {\em Parallel Algorithms, Proceedings of a {DIMACS} Workshop,
  Brunswick, New Jersey, USA, October 17-18, 1994}, pages 43--58, 1994.

\bibitem{DBLP:conf/tamc/HoraKT19}
Martin Hora, V{\'{a}}clav Koncick{\'{y}}, and Jakub Tetek.
\newblock Theoretical model of computation and algorithms for fpga-based
  hardware accelerators.
\newblock In T.~V. Gopal and Junzo Watada, editors, {\em Theory and
  Applications of Models of Computation - 15th Annual Conference, {TAMC} 2019,
  Kitakyushu, Japan, April 13-16, 2019, Proceedings}, volume 11436 of {\em
  Lecture Notes in Computer Science}, pages 295--312. Springer, 2019.

\bibitem{DBLP:journals/cacm/HuangL88}
Bing{-}Chao Huang and Michael~A. Langston.
\newblock Practical in-place merging.
\newblock {\em Commun. {ACM}}, 31(3):348--352, 1988.

\bibitem{DBLP:conf/spaa/KangGBD0M21}
Hongbo Kang, Phillip~B. Gibbons, Guy~E. Blelloch, Laxman Dhulipala, Yan Gu, and
  Charles McGuffey.
\newblock The processing-in-memory model.
\newblock In Kunal Agrawal and Yossi Azar, editors, {\em {SPAA} '21: 33rd {ACM}
  Symposium on Parallelism in Algorithms and Architectures, Virtual Event, USA,
  6-8 July, 2021}, pages 295--306. {ACM}, 2021.

\bibitem{DBLP:conf/spaa/KaufmannRS92}
Michael Kaufmann, Sanguthevar Rajasekaran, and Jop~F. Sibeyn.
\newblock Matching the bisection bound for routing and sorting on the mesh.
\newblock In {\em Proceedings of the 4th Annual {ACM} Symposium on Parallel
  Algorithms and Architectures, {SPAA} '92, San Diego, CA, USA, June 29 - July
  1, 1992}, pages 31--40, 1992.

\bibitem{DBLP:journals/micro/KecklerDKGG11}
Stephen~W. Keckler, William~J. Dally, Brucek Khailany, Michael Garland, and
  David Glasco.
\newblock Gpus and the future of parallel computing.
\newblock {\em {IEEE} Micro}, 31(5):7--17, 2011.

\bibitem{DBLP:journals/cse/KoggeS13}
Peter~M. Kogge and John Shalf.
\newblock Exascale computing trends: Adjusting to the "new normal"' for
  computer architecture.
\newblock {\em Comput. Sci. Eng.}, 15(6):16--26, 2013.

\bibitem{DBLP:journals/computer/Kung82}
H.~T. Kung.
\newblock Why systolic architectures?
\newblock {\em Computer}, 15(1):37--46, 1982.

\bibitem{Leighton1991IntroductionTP}
Frank~Thomson Leighton.
\newblock Introduction to parallel algorithms and architectures: Arrays, trees,
  hypercubes.
\newblock 1991.

\bibitem{Leighton2003ComplexityII}
Frank~Thomson Leighton.
\newblock Complexity issues in vlsi: Optimal layouts for the shuffle-exchange
  graph and other networks.
\newblock 2003.

\bibitem{DBLP:conf/parle/LeppanenP94}
Ville Lepp{\"{a}}nen and Martti Penttonen.
\newblock Simulation of {PRAM} models on meshes.
\newblock In {\em {PARLE} '94: Parallel Architectures and Languages Europe, 6th
  International {PARLE} Conference, Athens, Greece, July 4-8, 1994,
  Proceedings}, pages 146--158, 1994.

\bibitem{DBLP:journals/njc/LeppanenP95}
Ville Lepp{\"{a}}nen and Martti Penttonen.
\newblock Work-optimal simulation of {PRAM} models on meshes.
\newblock {\em Nord. J. Comput.}, 2(1):51--69, 1995.

\bibitem{DBLP:conf/stoc/LiptonS81}
Richard~J. Lipton and Robert Sedgewick.
\newblock Lower bounds for {VLSI}.
\newblock In {\em Proceedings of the 13th Annual {ACM} Symposium on Theory of
  Computing, May 11-13, 1981, Milwaukee, Wisconsin, {USA}}, pages 300--307.
  {ACM}, 1981.

\bibitem{liu2021closing}
Yong Liu, Xin Liu, Fang Li, Haohuan Fu, Yuling Yang, Jiawei Song, Pengpeng
  Zhao, Zhen Wang, Dajia Peng, Huarong Chen, et~al.
\newblock Closing the" quantum supremacy" gap: achieving real-time simulation
  of a random quantum circuit using a new sunway supercomputer.
\newblock In {\em Proceedings of the International Conference for High
  Performance Computing, Networking, Storage and Analysis}, pages 1--12, 2021.

\bibitem{66MortonZorder}
Guy~Macdonald Morton.
\newblock A computer oriented geodetic data base and a new technique in file
  sequencing.
\newblock 1966.

\bibitem{DBLP:journals/vldb/MullerTA12}
Ren{\'{e}} M{\"{u}}ller, Jens Teubner, and Gustavo Alonso.
\newblock Sorting networks on fpgas.
\newblock {\em {VLDB} J.}, 21(1):1--23, 2012.

\bibitem{DBLP:journals/corr/abs-2012-03112}
Onur Mutlu, Saugata Ghose, Juan G{\'{o}}mez{-}Luna, and Rachata
  Ausavarungnirun.
\newblock A modern primer on processing in memory.
\newblock {\em CoRR}, abs/2012.03112, 2020.

\bibitem{DBLP:journals/algorithmica/Paterson90}
Mike Paterson.
\newblock Improved sorting networks with o(log {N)} depth.
\newblock {\em Algorithmica}, 5(1):65--92, 1990.

\bibitem{ReifSynthesis}
John~H. Reif.
\newblock {\em Synthesis of Parallel Algorithms}.
\newblock Morgan Kaufmann Publishers Inc., San Francisco, CA, USA, 1st edition,
  1993.

\bibitem{DBLP:journals/siamcomp/Reischuk85}
R{\"{u}}diger Reischuk.
\newblock Probabilistic parallel algorithms for sorting and selection.
\newblock {\em {SIAM} J. Comput.}, 14(2):396--409, 1985.

\bibitem{DBLP:journals/taco/SanchezMK10}
Daniel S{\'{a}}nchez, George Michelogiannakis, and Christos Kozyrakis.
\newblock An analysis of on-chip interconnection networks for large-scale chip
  multiprocessors.
\newblock {\em {ACM} Trans. Archit. Code Optim.}, 7(1):4:1--4:28, 2010.

\bibitem{DBLP:journals/jcss/Savage81}
John~E. Savage.
\newblock Area-time tradeoffs for matrix multiplication and related problems in
  {VLSI} models.
\newblock {\em J. Comput. Syst. Sci.}, 22(2):230--242, 1981.

\bibitem{10.1007/BF02165411}
Volker Strassen.
\newblock Gaussian elimination is not optimal.
\newblock {\em Numer. Math.}, 13(4):354--356, aug 1969.

\bibitem{Cerebras}
Cerebras Systems{, Inc.}
\newblock Cerebras systems: Achieving industry bestai performance through a
  systems approach, April 2021.

\bibitem{DBLP:journals/tc/Thompson83a}
Clark~D. Thompson.
\newblock The {VLSI} complexity of sorting.
\newblock {\em {IEEE} Trans. Computers}, 32(12):1171--1184, 1983.

\bibitem{DBLP:journals/siamcomp/Valiant75}
Leslie~G. Valiant.
\newblock Parallelism in comparison problems.
\newblock {\em {SIAM} J. Comput.}, 4(3):348--355, 1975.

\bibitem{DBLP:journals/cacm/Valiant90}
Leslie~G. Valiant.
\newblock A bridging model for parallel computation.
\newblock {\em Commun. {ACM}}, 33(8):103--111, 1990.

\bibitem{xu2017benchmarking}
Zhigeng Xu, James Lin, and Satoshi Matsuoka.
\newblock Benchmarking sw26010 many-core processor.
\newblock In {\em 2017 IEEE International Parallel and Distributed Processing
  Symposium Workshops (IPDPSW)}, pages 743--752. IEEE, 2017.

\bibitem{DBLP:journals/micro/ZarubaSB21}
Florian Zaruba, Fabian Schuiki, and Luca Benini.
\newblock Manticore: {A} 4096-core {RISC-V} chiplet architecture for
  ultraefficient floating-point computing.
\newblock {\em {IEEE} Micro}, 41(2):36--42, 2021.

\end{thebibliography}

\end{document}